\documentclass[11pt,twoside]{article}
\usepackage{fancyhdr}
\usepackage[colorlinks,citecolor=blue,urlcolor=blue,linkcolor=blue,bookmarks=false]{hyperref}
\usepackage{amsfonts,epsfig,graphicx}
\usepackage{afterpage}
\usepackage{amsmath,amssymb,amsthm} 
\usepackage{fullpage}
\usepackage[T1]{fontenc} 
\usepackage{epsf} 
\usepackage{graphics} 
\usepackage{amsfonts,amsmath}
\usepackage{comment}
\usepackage[sort]{natbib} 
\usepackage{psfrag,xspace}
\usepackage{color,etoolbox}

\usepackage{subcaption}
\usepackage{booktabs}

\def\indist{\rightsquigarrow}
\def\ind{\perp\!\!\!\perp}
\def\T{{ \mathrm{\scriptscriptstyle T} }}

\newcommand{\var}{\text{var}}
\newcommand{\cov}{\text{cov}}
\newcommand{\Pb}{\mathbb{P}}

\newcommand{\Pn}{\mathbb{P}_n}

\newcommand{\E}{\mathbb{E}}
\newcommand{\R}{\mathbb{R}}

\newcommand{\bO}{\mathbf{O}}
\newcommand{\bo}{\mathbf{o}}

\newcommand{\bX}{\mathbf{X}}
\newcommand{\bx}{\mathbf{x}}
\newcommand{\bY}{\mathbf{Y}}
\newcommand{\by}{\mathbf{y}}
\newcommand{\bR}{\mathbf{R}}
\newcommand{\bbeta}{\boldsymbol{\beta}}
\newcommand{\bpsi}{\boldsymbol{\psi}}
\newcommand{\bphi}{\boldsymbol{\phi}}

\def\black{\textcolor{red}}
\def\black{\textcolor{black}}

\DeclareSymbolFont{bbold}{U}{bbold}{m}{n}
\DeclareSymbolFontAlphabet{\mathbbold}{bbold}
\newcommand{\one}{\mathbbold{1}}
\newtheorem{theorem}{Theorem}
\newtheorem{lemma}{Lemma}
\newtheorem{corollary}{Corollary}
\newtheorem{proposition}{Proposition}
\theoremstyle{definition}

\theoremstyle{remark}

\begin{document}

\def\spacingset#1{\renewcommand{\baselinestretch}%
{#1}\small\normalsize} \spacingset{1}


  \title{ Efficient Nonparametric Causal Inference \\ with Missing Exposure Information}
  \author{\\ Edward H. Kennedy \\ \\
    Department of Statistics \& Data Science, \\
    Carnegie Mellon University \\ \\ 
    \texttt{edward@stat.cmu.edu} \\
    }
  \maketitle
  \thispagestyle{empty}

\begin{abstract}
Missing exposure information is a very common feature of many observational studies.  Here we study identifiability and efficient estimation of causal effects on vector outcomes, in such cases where treatment is unconfounded but partially missing. We consider a missing at random setting where missingness in treatment can depend not only on complex covariates, but also on post-treatment outcomes. We give a new identifying expression for average treatment effects in this setting, along with the efficient influence function for this parameter in a nonparametric model, which yields a nonparametric efficiency bound. We use this latter result to construct nonparametric estimators that are less sensitive to the curse of dimensionality than usual, e.g., by having faster rates of convergence than the complex nuisance estimators they rely on. Further we show that these estimators can be root-n consistent and asymptotically normal under weak nonparametric conditions, even when constructed using flexible machine learning. Finally we apply these results to the problem of causal inference with a partially missing instrumental variable.
\end{abstract}

\section{Introduction}

It is very common for there to be missing data in observational studies where causal effects are of interest. In this paper we consider  studies where there is substantial missingness in an exposure variable. This is a very common feature of observational studies, and examples abound in the literature. For example, \citet{zhang2016causal} described the Consortium on Safe Labor observational study, where the goal was to estimate effects of mothers' body mass index on infant birthweight. There, covariate and outcome information was essentially always available, but body mass index data was only available for about half of the mothers. \citet{shortreed2010missing} used the Framingham Heart Study data to assess effects of physical activity on cardiovascular disease and mortality, but up to 30\% of subjects were missing physical activity information. Similarly, \citet{ahn2011missing} used the Molecular Epidemiology of Colorectal Cancer study to estimate effects of physical activity on colorectal cancer stage, but 20\% of subjects were missing physical activity data. \citet{shardell2014statistical} described an analysis of the Baltimore Hip Studies involving older adults with hip fractures, where the goal was to assess effects of perceived mobility recovery on survival. However, this self-reported mobility measure was unavailable for 27\% of subjects. \citet{molinari2010missing} and \citet{mebane2013causal} give numerous other examples of studies with missing exposure information,  particularly in survey settings, e.g., from the National Longitudinal Survey of Youth, and the Health and Retirement Study. This is certainly a prevalent problem. 

In fact the problem is even more widespread, since in instrumental variable studies one can view the instrument as a type of exposure (e.g., for the purpose of estimating intention-to-treat-style effects, as well as other instrumental variable estimands that require estimating instrument effects on both treatment and outcome). And it is similarly common for instrument values to be missing. For example, in a Mendelian randomization context \citet{burgess2011missing} used genetic variants as instrumental variables to study the effect of C-reactive protein on fibrinogen and coronary heart disease. However, data on these variants was missing for up to 10\% of subjects, due to difficulty in interpreting output from genotyping platforms. \citet{mogstad2012instrumental} and \citet{chaudhuri2016gmm} give further examples of missing instruments from economics. 

Although missing exposures and instruments are a prevalent problem, the proposed methods for dealing with this issue have relied on potentially restrictive modeling assumptions, and have been somewhat ad hoc in not considering optimal efficiency. For example \citet{williamson2012doubly} and \citet{zhang2016causal} propose interesting semiparametric estimators, but they rely on parametric models for nuisance functions, and do not consider the question of efficiency in either nonparametric or semiparametric models. \citet{zhang2016causal}  also only considers binary outcomes. \citet{chaudhuri2016gmm} discusses (semiparametric) efficiency theory, but only for a finite-dimensional parameter in a population moment condition depending on a known function. This means their results apply to classical linear models, but not to the fully nonparametric setting pursued here. \citet{kennedy2017paradoxes} consider nonparametric efficiency theory in missing instrumental variable problems, but only in simpler settings with one-sided noncompliance and no covariates. 

Thus we fill these gaps by giving a new identifying expression for average treatment effects \black{of multivalued discrete exposures} in the presence of complex confounding and missing exposure values, deriving the efficient influence function and corresponding nonparametric efficiency bounds, and constructing nonparametric estimators that can be $\sqrt{n}$-consistent and asymptotically normal, even if nuisance functions are estimated at slower rates via nonparametric machine learning tools. Finally we apply these general results to also address the problem of causal inference with a partially missing instrumental variable. Throughout we make use of a missing at random assumption used by previous authors, allowing exposure missingness to depend on post-exposure outcome information.

\section{Missing Exposures}

In this section we consider the general problem of identification and efficient estimation of average treatment effects, when exposure values are missing at random, allowing the missingness mechanism to depend on both covariates and post-treatment outcome information.

\subsection{Setup}

Suppose we observe an iid sample $(\bO_1,...,\bO_n) \sim \Pb$ where
$$ \bO = (\bX, R, RZ, \bY) $$
for $\bX \in \R^d$ denoting covariate information, $Z \in \{z_1,...,z_k\}$ a discrete treatment or exposure, $R \in \{0,1\}$ an indicator for whether $Z$ is observed or not, and $\bY = (Y_1,...,Y_p)^\T \in \R^p$ a vector of $p$ outcomes of interest. In general we use script characters to denote the support of a random variable, e.g., $\bX \in \mathcal{X} \subseteq \R^d$. For notational simplicity we further define the nuisance functions
\begin{align*}
\mu(\by \mid \bx) &= \Pb(\bY \leq \by \mid \bX=\bx) \\
\pi(\bx,\by) &= \Pb(R=1 \mid \bX=\bx, \bY=\by) \\
\lambda_z(\bx,\by) &= \Pb(Z=z \mid \bX=\bx,\bY=\by, R=1) .
\end{align*}
Note that $\mu$ is the cumulative distribution function of the outcome given covariates, $\pi$ can be viewed as the missingness propensity score,  and $\lambda$ the regression on covariates and outcomes of treatment among those for whom it is measured. We further define
\begin{align*}
\bbeta_z(\bx) &= \E\{ \bY \lambda_z(\bX, \bY) \mid \bX=\bx\} = \int_\mathcal{Y} \by \lambda_z(\bx,\by) \ d\mu(\by \mid \bx) \\
\gamma_z(\bx) &= \E\{ \lambda_z(\bX, \bY) \mid \bX=\bx\} = \int_\mathcal{Y} \lambda_z(\bx, \by) \ d\mu(\by \mid \bx) .
\end{align*}
The quantity $\bbeta_z(\bx) = \{ \beta_z^1(\bx),...,\beta_z^p(\bx) \}^\T$ is a vector of the same dimension as $\bY$. We will see shortly that, under missing at random assumptions, $\gamma_z$ equals the propensity score, while $\bbeta_z$ equals the product of the propensity score and outcome regression. 

Our goal is to estimate the mean $\bpsi_z = \E(\bY^z) = \{ \E(Y_1^z), ..., \E(Y_p^z) \}^\T \in \R^p$ of the outcomes that would have been observed under treatment level $z \in \mathcal{Z}$. It is well-known that this equals
\begin{equation} \label{eq:psi2}
\bpsi_z = \int_{\mathcal{X}} \E(\bY \mid \bX=\bx, Z=z) \ d\Pb(\bx) 
\end{equation}
under the following standard causal assumptions:
\begin{align}
\text{(Consistency.) } \ & \bY=  \bY^Z \tag{A1}  \label{A1} \\
\text{($Z$-Positivity.) } \ & \Pb\{ \epsilon < \Pb(Z=z \mid \bX) < 1-\epsilon\} = 1 \text{ for all } z \in \mathcal{Z} \tag{A2}  \label{A2} \\
\text{($Z$-Exchangeability.) } \ & Z \ind \bY^z \mid \bX \text{ for all } z \in \mathcal{Z} \tag{A3} \label{A3}
\end{align}
These assumptions have been discussed extensively in the literature \citep{imbens2004nonparametric, van2003unified}, so we refer the reader elsewhere for more details.  

Crucially, when treatment $Z$ is missing for some subjects, expression \eqref{eq:psi2} is still not identified even under \eqref{A1}--\eqref{A3}, since $Z$ is not observed unless $R=1$. We consider identification under missing at random conditions used for example by \citet{williamson2012doubly, chaudhuri2016gmm} and \citet{zhang2016causal}, which are:
\begin{align}
\text{($R$-Exchangeability.) } \ & R \ind Z \mid \bX, \bY \tag{A4} \label{A4} \\
\text{($R$-Positivity.) } \ & \Pb\{ \epsilon < \Pb(R=1 \mid \bX, \bY) < 1-\epsilon\} = 1 \tag{A5}  \label{A5}
\end{align}
Note that the missing at random condition \eqref{A4} allows missingness in treatment $Z$ to depend on the post-treatment outcome $\bY$; this will be important if the outcome captures some information about the missingness mechanism beyond the covariates. An alternative missing-at-random assumption would be $R \ind (Z,\bY) \mid \bX$; note however that this implies our $R$-exchangeability assumption, as well as the further testable implication that $R \ind Y \mid \bX$. Therefore our assumption is strictly weaker. 

\black{Figure \ref{dags} uses directed acyclic graphs to illustrate two different data generating processes that satisfy exchangeability conditions \eqref{A3} and \eqref{A4}. The first represents a process where missingness occurs prior to the outcome (e.g., subjects miss a visit when they would have contributed treatment information); the second represents a process where missingness occurs after the outcome (e.g., survey non-response or data corruption after measurement). }

\begin{figure}[h!]
\centering
\subcaptionbox{\label{pre}}
{\includegraphics[width=0.45\textwidth]{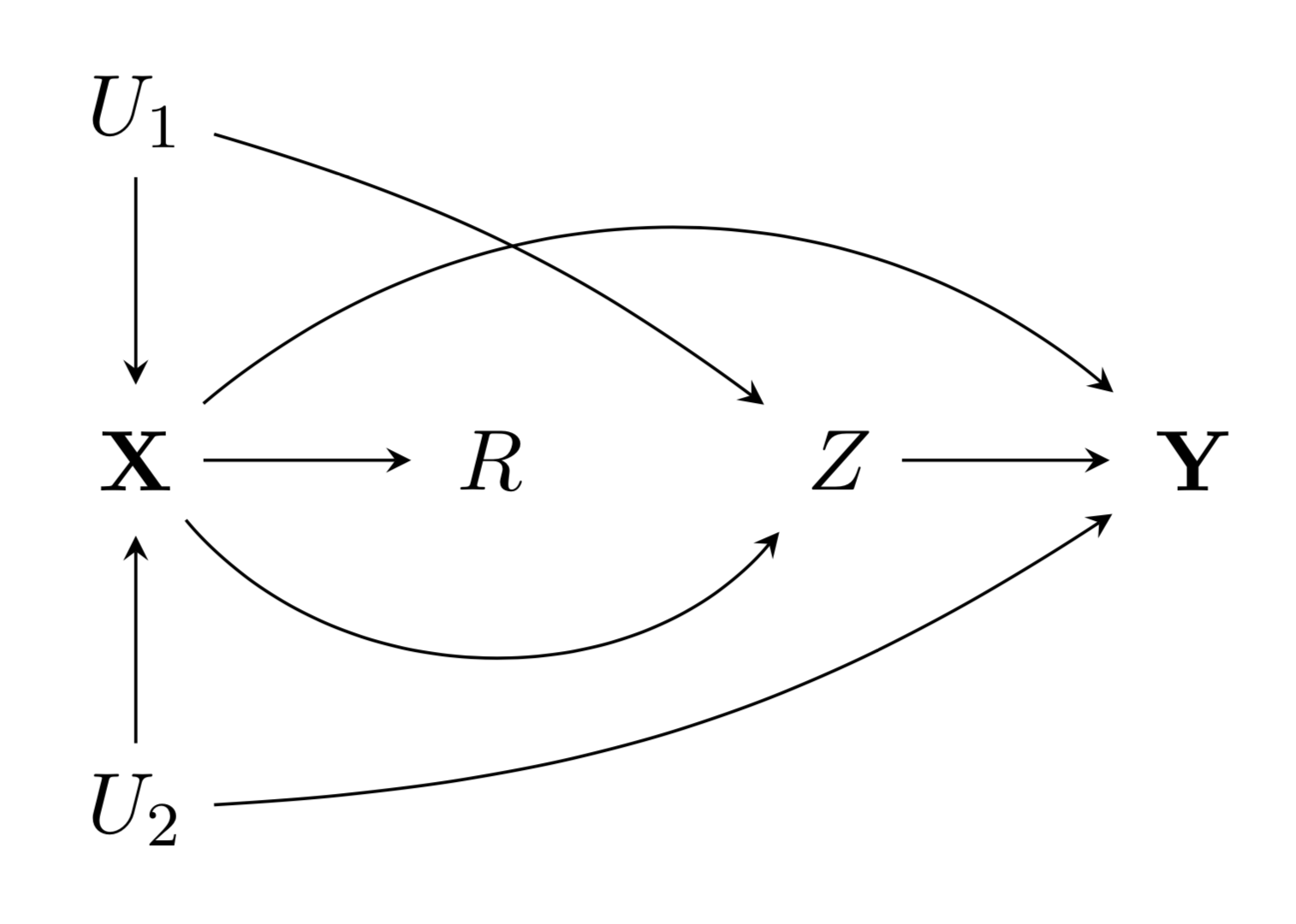}}
\subcaptionbox{\label{post}}
{\includegraphics[width=0.45\textwidth]{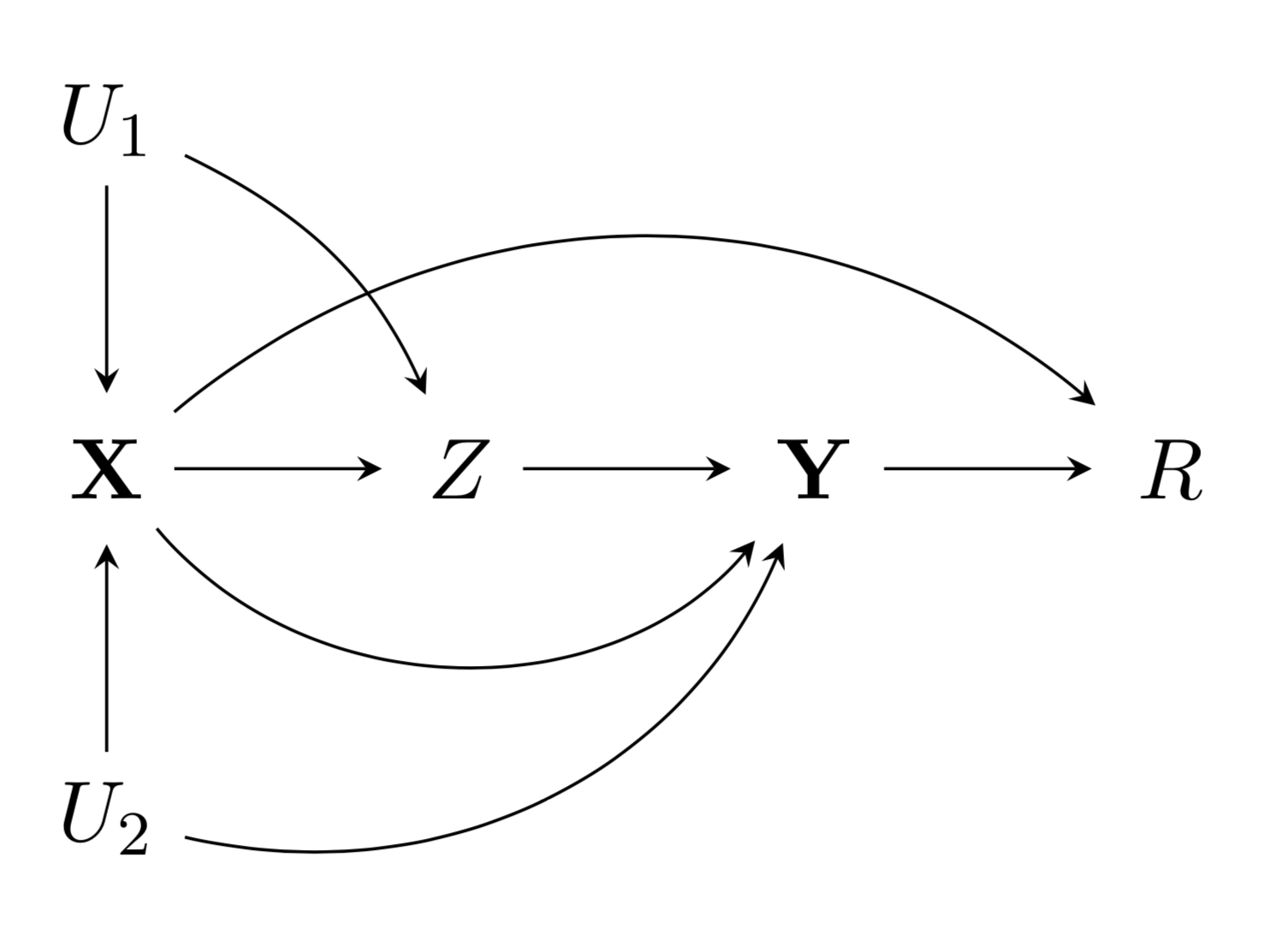}}
\caption{Two directed acyclic graphs for which the required exchangeability assumptions \eqref{A3} and \eqref{A4} hold, where $(U_1,U_2)$ are unmeasured and $Z$ is only observed when $R=1$. In graph (a) missingness can be viewed as occurring prior to the outcome, while in (b) it can be viewed as occurring after. The variable $U_2$ can be represented as the potential outcome $\mathbf{Y}^0$.}\label{dags}
\end{figure}

\subsection{Identification \& Efficiency Theory}

Our first result gives a new identifying expression for $\bpsi_z$ under the causal and missing at random assumptions above. This essentially follows from the important facts that, under \eqref{A4}--\eqref{A5}, the propensity score is given by
$$ \gamma_z(\bx) = \Pb(Z=z \mid \bX=\bx) $$
(note this means $Z$-positivity \eqref{A2} is equivalent to $\gamma_z$ being bounded away from zero and one), and that the outcome regression satisfies
$$ \bbeta_z(\bx)  =  \gamma_z(\bx) \E(\bY \mid \bX=\bx, Z=z)  .  $$

\begin{proposition} \label{ident}
Under the causal assumptions \eqref{A1}--\eqref{A3} and the missing at random assumptions \eqref{A4}--\eqref{A5}, it follows that
$$ \bpsi_z = \E(\bY^z) = \E\left\{ \frac{\bbeta_z(\bX)}{\gamma_z(\bX)} \right\} . $$
\end{proposition}
\begin{proof}
We have
\begin{align*}
\E(\bY^z) &= \int_{\mathcal{X}} \E(\bY \mid \bX=\bx, Z=z) \ d\Pb(\bx) \\
&= \int_\mathcal{X} \int_\mathcal{Y} \frac{\by \Pb(Z=z \mid \bX=\bx, \bY=\by)}{\Pb(Z=z \mid \bX=\bx)} \ d\mu(\by \mid \bx) \ d\Pb(\bx) \\
&= \int_\mathcal{X} \int_\mathcal{Y} \frac{\by \ \lambda_z(\bx, \by) }{\int_\mathcal{Y} \lambda_z(\bx, \by) d\mu(\by \mid \bx)} \ d\mu(\by \mid \bx) \ d\Pb(\bx) = \E\left\{ \frac{\bbeta_z(\bX)}{\gamma_z(\bX)} \right\} 
\end{align*}
where the first equality follows by the causal assumptions \eqref{A1}--\eqref{A3}, the second by Bayes' rule, the third by the missing at random assumptions \eqref{A4}--\eqref{A5} and iterated expectation, and the fourth by definition. 
\end{proof}

Interestingly, although the complete-data functional \eqref{eq:psi2} does not depend on the observational treatment process, its identified version under the missing at random assumptions does. Intuitively, this occurs because when $Z$ is missing, one cannot simply condition on $(\bX,Z)$ anymore, and instead the distribution of $Z$ given $\bX$ needs to be constructed by marginalizing over that of $Z$ given $(\bX, \bY)$ among those with $Z$ observed. 

The next result gives a crucial von Mises-type expansion for the parameter from Proposition \ref{ident}, which lays the foundation for the efficiency bound and estimation results to come. This result can be viewed as giving a distributional Taylor expansion for the functional $\bpsi_z$.

\begin{lemma} \label{vonmises}
The functional $\bpsi_z=\bpsi_z(\Pb)$ from Proposition \ref{ident} admits the expansion
$$ \bpsi_z(\overline\Pb) - \bpsi_z(\Pb) = \int_\mathcal{O} \Big\{ \bphi_z(\bO; \overline\Pb) - \bpsi_z(\overline\Pb) \Big\} (d\overline\Pb - d\Pb) + \bR_z(\overline\Pb, \Pb) $$
where
$$ \bphi_z(\bO; \Pb) = \left\{ \frac{\bY-\bbeta_z(\bX)/\gamma_z(\bX)}{\gamma_z(\bX)} \right\} \left[ \frac{R \{ \one(Z=z) - \lambda_z(\bX,\bY) \} }{\pi(\bX,\bY)} + \lambda_z(\bX,\bY) \right]  + \frac{\bbeta_z(\bX)}{\gamma_z(\bX)} $$
and
\begin{align*}
\bR_z(\overline\Pb, \Pb) = \E_\Pb &\bigg( \left\{ \frac{\bY-\overline\bbeta_z(\bX)/\overline\gamma_z(\bX)}{\overline\gamma_z(\bX)}  \right\} \left\{ \frac{\pi(\bX,\bY) - \overline\pi(\bX,\bY)}{\overline\pi(\bX,\bY)} \right\} \Big\{ \lambda_z(\bX,\bY) - \overline\lambda_z(\bX,\bY) \Big\}  \\
& \hspace{.3in} + \left[ \left\{ \frac{\bbeta_z(\bX)-\overline\bbeta_z(\bX)}{\gamma_z(\bX)}  \right\} + \frac{\overline\bbeta_z(\bX)}{ \overline\gamma_z(\bX) } \left\{ \frac{ \overline\gamma_z(\bX) - \gamma_z(\bX) }{\gamma_z(\bX)} \right\} \right] \left\{ \frac{\gamma_z(\bX) - \overline\gamma_z(\bX)}{\overline\gamma_z(\bX)} \right\} \bigg) .
\end{align*}
\end{lemma}
\begin{proof}
Here we drop the $z$ argument throughout to ease notation. Note we can write 
\begin{align*}
\bphi(\bO; \Pb) &= \left[ \frac{R \{ \one(Z=z) - \lambda(\bX,\bY) \} }{\pi(\bX,\bY)} + \lambda(\bX,\bY) \right] \left\{ \frac{\bY-\bbeta(\bX)/\gamma(\bX)}{\gamma(\bX)} \right\} + \frac{\bbeta(\bX)}{\gamma(\bX)}  \\
&= \frac{1}{\gamma(\bX)} \left[ \frac{R\bY}{\pi(\bX,\bY)} \Big\{ \one(Z=z) - \lambda(\bX,\bY) \Big\} + \Big\{ \bY \lambda(\bX,\bY) - \bbeta(\bX) \Big\} \right] \\
&\hspace{.3in} - \frac{\bbeta(\bX)}{ \gamma(\bX)^2} \left[ \frac{R \{ \one(Z=z) - \lambda(\bX,\bY) \} }{\pi(\bX,\bY)} + \Big\{ \lambda(\bX,\bY)-\gamma(\bX) \Big\} \right] + \frac{\bbeta(\bX)}{\gamma(\bX)} .
\end{align*}
Therefore, letting $\E=\E_\Pb$ and dropping $(\bX,\bY)$ arguments from $(\pi,\lambda,\bbeta,\gamma)$ to further ease notation, we have 
\begin{align*}
 \bpsi(\overline\Pb) &- \bpsi(\Pb) + \int_\mathcal{O} \Big\{ \bphi(\bO; \overline\Pb) - \bpsi(\overline\Pb) \Big\} d\Pb \\
&= \E\bigg[ \frac{1}{\overline\gamma} \left\{ \frac{R \bY}{\overline\pi } \Big( \one_z - \overline\lambda \Big) + \Big( \bY \overline\lambda - \overline\bbeta \Big) \right\}  
- \frac{\overline\bbeta}{ \overline\gamma^2} \left\{ \frac{R ( \one_z - \overline\lambda ) }{\overline\pi} +  \Big( \overline\lambda -\overline\gamma \Big) \right\} 
+ \left( \frac{\overline\bbeta}{\overline\gamma} - \frac{\bbeta}{\gamma} \right) \bigg] \\
&= \E\bigg[ \left( \frac{\bY - \overline\bbeta/\overline\gamma}{\overline\gamma} \right) \left( \frac{\pi - \overline\pi}{\overline\pi } \right) \Big( \lambda - \overline\lambda \Big) 
+ \left( \frac{\bbeta - \overline\bbeta}{\overline\gamma} \right) - \frac{\overline\bbeta}{\overline\gamma} \left( \frac{\gamma - \overline\gamma}{\overline\gamma} \right) + \left( \frac{\overline\bbeta}{\overline\gamma} - \frac{\bbeta}{\gamma} \right) \bigg] \\
&= \E\bigg[ \left( \frac{\bY - \overline\bbeta/\overline\gamma}{\overline\gamma} \right) \left( \frac{\pi - \overline\pi}{\overline\pi } \right) \Big( \lambda - \overline\lambda \Big) 
+ \left( \frac{\bbeta}{\gamma} - \frac{\overline\bbeta}{\overline\gamma} \right) \left( \frac{\gamma - \overline\gamma}{\overline\gamma} \right) \bigg] = \bR_z(\overline\Pb,\Pb)
\end{align*}
where the second equality follows from from rearranging and iterated expectation (together with $\E(\bY \lambda \mid \bX)=\bbeta$ and $\E(\lambda \mid \bX)=\gamma$), and the third since 
$$
\frac{\bbeta - \overline\bbeta}{\overline\gamma}  - \frac{\overline\bbeta}{\overline\gamma} \left( \frac{\gamma - \overline\gamma}{\overline\gamma} \right) + \frac{\overline\bbeta}{\overline\gamma} - \frac{\bbeta}{\gamma}  
= \bbeta \left( \frac{1}{\overline\gamma} - \frac{1}{\gamma} \right) - \frac{\overline\bbeta}{\overline\gamma} \left( \frac{\gamma - \overline\gamma}{\overline\gamma} \right) 
 = \left( \frac{\bbeta}{\gamma} - \frac{\overline\bbeta}{\overline\gamma} \right) \left( \frac{\gamma - \overline\gamma}{\overline\gamma} \right) .
$$
Note the above implies $\E\{ \bphi(\bO;\Pb)-\bpsi_z(\Pb)\}=0$, which is also straightforward to see using iterated expectation. 
\end{proof}

Lemma \ref{vonmises} has several important consequences. First, it suggests how one could correct the first-order bias of a plug-in estimator $\bpsi_z(\widehat\Pb)$, by estimating the first term in the expansion and subtracting it off. This is one way to view what semiparametric estimators (particularly of the ``one-step'' variety) based on influence functions are doing, and in fact the estimator presented in the next subsection does precisely this. Second, it essentially immediately yields the efficient influence function for $\bpsi_z$. The next theorem states this result; after it we describe why the efficient influence function is useful here.

\begin{theorem} \label{eif}
Under a nonparametric model satisfying positivity conditions \eqref{A2} and \eqref{A5}, the efficient influence function for $\bpsi_z$ is given by $\bphi_z(\bO; \Pb)-\bpsi_z$ as defined in Lemma \ref{vonmises}.
\end{theorem}
\begin{proof}
Recall from \citet{bickel1993efficient} and \citet{van2002semiparametric} that the nonparametric efficiency bound for a functional $\psi$ is given by the supremum of Cramer-Rao lower bounds for that functional across smooth parametric submodels. The efficient influence function is the mean-zero function whose variance equals the efficiency bound, and is given by the unique $\varphi$ that is a valid submodel score (or limit of such scores) satisfying pathwise differentiability, i.e.,
\begin{equation} \label{eq:pathwise}
 \frac{d}{d\epsilon} \psi(\Pb_\epsilon) \Bigm|_{\epsilon=0} = \int_\mathcal{O} \varphi(\bO;\Pb) \left( \frac{d}{d\epsilon} \log d\Pb_\epsilon \right) \Bigm|_{\epsilon=0} d\Pb  
 \end{equation}
for $\Pb_\epsilon$ any smooth parametric submodel. 

To see that $\bphi - \bpsi$ is the efficient influence function for $\bpsi$, first note that the expansion in Lemma \ref{vonmises} implies
$$  \bpsi_z(\Pb_\epsilon) - \bpsi_z(\Pb) =  \int_\mathcal{O} \Big\{ \bphi_z(\bO; \Pb) - \bpsi_z(\Pb) \Big\} d\Pb_\epsilon - \bR_z(\Pb, \Pb_\epsilon) $$
so differentiating with respect to $\epsilon$ yields
\begin{align*}
\frac{d}{d\epsilon} \bpsi_z(\Pb_\epsilon) &=  \int_\mathcal{O} \Big\{ \bphi_z(\bO; \Pb) - \bpsi_z(\Pb) \Big\} \frac{d}{d\epsilon}  d\Pb_\epsilon - \frac{d}{d\epsilon} \bR_z(\Pb, \Pb_\epsilon) \\
&= \int_\mathcal{O} \Big\{ \bphi_z(\bO; \Pb) - \bpsi_z(\Pb) \Big\} \left( \frac{d}{d\epsilon} \log  d\Pb_\epsilon \right) d\Pb_\epsilon - \frac{d}{d\epsilon} \bR_z(\Pb, \Pb_\epsilon) .
\end{align*}
The property \eqref{eq:pathwise} follows after evaluating at $\epsilon=0$, since 
$$ \frac{d}{d\epsilon} \bR_z(\Pb, \Pb_\epsilon) \Bigm|_{\epsilon=0} = 0 $$
by virtue of the fact that $\bR_z(\Pb, \Pb_\epsilon)$ consists of only second-order products of errors between $\Pb_\epsilon$ and $\Pb$. Thus applying the product rule yields a sum of two terms, each of which is a product of a derivative term (which may not be zero at $\epsilon=0$) and an error term involving differences of components of $\Pb_\epsilon$ and $\Pb$ (which will be zero at $\epsilon=0$). Since our model is nonparametric, the tangent space is the entire Hilbert space of mean-zero finite-variance functions; hence there is only one influence function satisfying \eqref{eq:pathwise} and it is the efficient one \citep{bickel1993efficient, van2002semiparametric, tsiatis2006semiparametric}. Therefore $\bphi - \bpsi$ must be the efficient influence function. 

An equivalent way to derive this result, as suggested by an anonymous reviewer in a previous version of this manuscript, is to use results from \citet{robins1994estimation}. Specifically, as in Theorem 7.2 of \citet{tsiatis2006semiparametric}, one can take full-data efficient influence function for $\bpsi_z$, inverse-probability-weight it for those with $R=1$ and subtract off its projection onto the tangent space. This yields the same efficient influence function.
\end{proof}

The efficient influence function is important since its variance $\cov\{\bphi_z(\bO; \Pb)-\bpsi_z\}$ gives an efficiency bound for estimation of $\bpsi_z$, providing a benchmark for efficient estimation. More precisely, following \citet{bickel1993efficient}, \citet{van2002semiparametric}, and \citet{tsiatis2006semiparametric}, this variance provides a local asymptotic minimax lower bound in the sense of Hajek and Le Cam, and tells us that the asymptotic variance of any regular asymptotically linear estimator can be no smaller (in that the difference in covariance matrices must be non-negative definite). Insofar as the bias-correction suggested earlier directly involves the efficient influence function, this object is also crucial for constructing estimators that have second-order bias and so can be $\sqrt{n}$-consistent and asymptotically normal even when the nuisance functions are estimated flexibly at slower rates of convergence. This feature will be detailed in the next subsection.

\subsection{Estimation \& Inference}

Here we present an estimator based on the functional expansion from Lemma \ref{vonmises}, which is asymptotically efficient under weak nonparametric conditions.

To ease notation let $\bphi_z = \bphi_z(\bO; \Pb)$ and $\widehat\bphi_z = \bphi_z(\bO; \widehat\Pb)$ denote the true and estimated versions of the uncentered efficient influence function for $\bpsi_z$. The estimator we study here is given by
$$ \widehat\bpsi_z = \Pn( \widehat\bphi_z ) $$
where we use $\Pn(f)=\frac{1}{n} \sum_{i=1}^n f(\bO_i)$ to denote sample averages.  Therefore the estimator $\widehat\bpsi_z$ is simply the sample average of the estimated (uncentered) influence function values; \black{equivalently we can write it as a bias-corrected version of the plug-in $\bpsi_z(\widehat\Pb)$, namely
$$ \widehat\bpsi_z = \bpsi_z(\widehat\Pb) + \Pn( \widehat{\boldsymbol{\varphi}}_z )  $$
where $\widehat{\boldsymbol{\varphi}}_z= \widehat\bphi_z - \bpsi_z(\widehat\Pb)$ is the estimated efficient influence function. }

For simplicity, in the following results we assume the nuisance estimates $\widehat\Pb$ are constructed from a separate independent sample. In practice, one can split the sample, use part for fitting $\widehat\Pb$ and the other for constructing $\widehat\bphi_z$, and then swap so as to attain full efficiency based on the entire sample size $n$ rather than a fraction, e.g., $n/2$. This is the idea behind the sample-splitting methods used in other functional estimation problems \citep{robins2008higher, zheng2010asymptotic, chernozhukov2016double}. Alternatively, if the same observations are used both for estimating $\widehat\Pb$ and constructing $\widehat\bphi_z$, one generally needs to rely on empirical process conditions to obtain the kinds of results we present here.

The next theorem gives the asymptotic properties of the estimator $\widehat\bpsi_z$, and conditions under which it is $\sqrt{n}$-consistent and converging to a normal distribution with asymptotic variance equal to the nonparametric efficiency bound.  In what follows, we let $\| f \|^2 = \Pb(f^2)= \int_\mathcal{O} f(\bo)^2 \ d\Pb(\bo)$ denote the squared $L_2(\Pb)$ norm.

\begin{theorem} \label{trtest}
Assume $\| \widehat\bphi_z - \bphi_z \| = o_\Pb(1)$ and $\Pb( \epsilon < \widehat\pi < 1-\epsilon)=\Pb(\epsilon < \widehat\gamma_z < 1-\epsilon)=1$. Then
$$ \widehat\bpsi_z - \bpsi_z = O_\Pb\left( \frac{1}{\sqrt{n}} + \| \widehat\pi - \pi \| \| \widehat\lambda_z - \lambda_z \| + \left( \| \widehat\bbeta_z - \bbeta_z \| + \| \widehat\gamma_z - \gamma_z \| \right) \| \widehat\gamma_z - \gamma_z \| \right) , $$
and if $\| \widehat\pi - \pi \| \| \widehat\lambda_z - \lambda_z \| + \left( \| \widehat\bbeta_z - \bbeta_z \| + \| \widehat\gamma_z - \gamma_z \| \right) \| \widehat\gamma_z - \gamma_z \|=o_\Pb(1/\sqrt{n})$, we have
$$ \sqrt{n}(\widehat\bpsi_z - \bpsi_z) \indist N\Big( \mathbf{0}, \cov(\bphi_z) \Big).  $$
\end{theorem}
\begin{proof}
Dropping $z$ subscripts to ease notation, we can write
\begin{equation} \label{eq:decomp}
\widehat\bpsi - \bpsi = (\Pn - \Pb) (\widehat\bphi - \bphi) + (\Pn-\Pb) \bphi + \Pb (\widehat\bphi - \bphi) .
\end{equation}
For the first term in \eqref{eq:decomp} above, Lemma \ref{splitlem} in the Appendix (reproduced from \citet{kennedy2019sharp}) implies that
$$ (\Pn - \Pb) (\widehat\bphi - \bphi) = O_\Pb\left( \frac{\| \widehat\bphi - \bphi \| }{\sqrt{n}} \right) = o_\Pb(1/\sqrt{n}) $$
where the last equality follows since $\| \widehat\bphi - \bphi \| = o_\Pb(1)$ by assumption. The expansion from Lemma \ref{vonmises} now implies
\begin{align*}
\Pb (\widehat\bphi - \bphi) &= \int_\mathcal{O} \Big\{ \bphi(\bO;\widehat\Pb) - \bpsi(\Pb) \Big\} d \Pb = \bR(\widehat\Pb,\Pb) \\
&\lesssim \| \widehat\pi - \pi \| \| \widehat\lambda - \lambda \| + \left( \| \widehat\bbeta - \bbeta \| + \| \widehat\gamma - \gamma \| \right) \| \widehat\gamma - \gamma \|
\end{align*}
where the last line uses Cauchy-Schwarz and the fact that $(\widehat\gamma,\widehat\pi,\gamma)$ are all bounded away from zero. This yields the result.
\end{proof}

Importantly, Theorem \ref{trtest} shows that $\widehat\bpsi_z$ attains faster rates than its nuisance estimators, and can be asymptotically efficient under weak nonparametric conditions. Specifically, as long as the influence function is consistently estimated in $L_2$ norm, the estimator $\widehat\bpsi_z$ has a rate of convergence that is second-order in the nuisance estimation error, thus attaining faster rates than the nuisance estimators. Under standard $n^{-1/4}$-type rate conditions, the estimator is $\sqrt{n}$-consistent, asymptotically normal, and efficient. Importantly, these rates can plausibly be attained under nonparametric smoothness, sparsity, or other structural conditions (e.g., additive modeling or bounded variation assumptions, etc.). \black{For example, if it is assumed that all $d$-dimensional nuisance functions lie in a H\"{o}lder class with smoothness index $s$ (i.e., partial derivatives up to order $s$ exist and are Lipschitz) then the assumption of Theorem \ref{trtest} would be satisfied when $s>d/2$, i.e., the smoothness index is at least half the dimension. Alternatively, if the functions are $s$-sparse then one would need $s = o( \sqrt{n})$ up to log factors, as in \citet{farrell2015robust}.} Then asymptotically valid 95\% confidence intervals can be constructed via a simple Wald form, $\widehat\bpsi_z \pm 1.96 \sqrt{\text{diag}\{\cov(\bphi_z)\}/n}$ The next result points out the double robustness of $\widehat\bpsi_z$.

\begin{corollary} \label{consistency}
Under the conditions of Theorem \ref{trtest}, the estimator $\widehat\bpsi_z$ is consistent if either
\begin{enumerate}
\item $\| \widehat\gamma_z - \gamma_z \|=o_\Pb(1)$ and $\| \widehat\pi - \pi \| = o_\Pb(1)$, or
\item $\| \widehat\gamma_z - \gamma_z \|=o_\Pb(1)$ and $\| \widehat\lambda_z - \lambda_z \|=o_\Pb(1)$.
\end{enumerate}
\end{corollary}

Corollary \ref{consistency} shows that $\widehat\bpsi_z$ is doubly robust \citep{scharfstein1999adjusting, robins2000robust}, since it is consistent if either $\widehat\pi$ or $\widehat\lambda_z$ are (and $\widehat\gamma_z$ is). Note however that our formulation requires the propensity score $\gamma_z$ to be estimated consistently. This contrasts with the semiparametric approach of \citet{zhang2016causal}, who construct an estimator that is consistent as long as two of three nuisance functions are estimated consistently. However, \citet{zhang2016causal} work under a different factorization of the likelihood, and impose parametric models on the partially observed propensity score and outcome regression functions. \black{It is unclear whether our remainder can be written in a triply robust form, though we conjecture that results of \citet{zhang2016causal} would not hold in the fully nonparametric setting considered here. This and a more general study of triple robustness could be an important avenue for future work.}

\section{Application to Missing Instruments}

Here we apply the theory from the previous section to identify and efficient estimate the local average treatment effect in instrumental variable studies with missing instrument values.

It is quite common for some instrument values to be missing in instrumental variable studies \citep{mogstad2012instrumental, chaudhuri2016gmm, kennedy2017paradoxes}. This setup fits in the proposed framework from the previous section as follows. We have $\bO = (\bX, R, RZ, \bY)$ where $Z \in \{0,1\}$ is now an instrument, and $\bY=(A,Y)$ for $A \in \{0,1\}$ a binary treatment and $Y \in \R$ an outcome of interest. \black{Here the outcome $Y \in \R$ is a scalar, but a multivariate outcome presents no additional complications.} Note also our slight abuse of notation in using bold $\bY=(A,Y)$ to denote a vector that contains the scalar outcome $Y$. Then we can write
$$ \bbeta_z(\bx) = \{ \beta_z^a(\bx) , \beta_z^y(\bx) \}^\T $$
where $\beta_z^t(\bx) = \E\{ T \lambda_z(\bX,A,Y) \mid \bX=\bx\}$. 

In addition to the causal assumptions \eqref{A1}--\eqref{A3} and missing at random assumptions \eqref{A4}--\eqref{A5} from before, we further make the instrumental variable assumptions:
\begin{align}
\text{(Exclusion.) } \ & Y^{za}=Y^a \tag{A6} \\
\text{(Relevance.) } \ & \Pb(A^{z=1} > A^{z=0}) > \epsilon \tag{A7} \\
\text{(Monotonicity.) } \ & \Pb(A^{z=1} \geq A^{z=0})=1 \tag{A8}
\end{align}

Our first result identifies the local average treatment effect under the assumptions above. 

\begin{proposition} \label{ivident}
Under the causal assumptions \eqref{A1}--\eqref{A3}, the missing at random assumptions \eqref{A4}--\eqref{A5}, and the instrumental variable assumptions (A6)--(A8), it follows that
$$ \theta = \E(Y^{a=1}-Y^{a=0} \mid A^{z=1}>A^{z=0}) =  \frac{ \E\{ \beta_1^y(\bX) / \gamma_1(\bX) - \beta_0^y(\bX) / \gamma_0(\bX)  \} }{ \E\{ \beta_1^a(\bX) / \gamma_1(\bX) - \beta_0^a(\bX) / \gamma_0(\bX)  \} } . $$
\end{proposition}
\begin{proof}
It is well known \citep{imbens1994identification, abadie2003semiparametric} that assumptions (A6)--(A8) imply
$$ \theta = \frac{ \E(Y^{z=1}-Y^{z=0})} { \E(A^{z=1}-A^{z=0})} $$
so the result follows from Proposition \ref{ident}, after taking $\bY=(A,Y)$ and $\mathcal{Z} = \{0,1\}$. 
\end{proof}

Although we focus on the local average treatment effect, the same observed data functional can represent other treatment effects under varying assumptions (e.g., the effect on the would-be-treated under a no-effect-modification assumption as discussed for example by \citet{hernan2006instruments}). Thus our results equally apply to these other settings.

Now we go on to use the theory from the previous section to construct an efficient estimator of the local average treatment effect $\theta$. As with $\bbeta_z(\bx)$, we can decompose the efficient influence function $\bphi_z$ from the previous section as
$$ \bphi_z(\bO) = \{ \phi_z^a(\bO) , \phi_z^y(\bO) \}^\T $$
for the two outcomes $(A,Y) \in \bY$. As before we write $\bphi_z = \bphi_z(\bO; \Pb)$ and $\widehat\bphi_z = \bphi_z(\bO; \widehat\Pb)$ to ease notation, and suppose $\widehat\Pb$ is constructed from an independent sample. The proposed estimator is given by
$$ \widehat\theta = \frac{ \Pn( \widehat\phi_1^y - \widehat\phi_0^y ) }{ \Pn( \widehat\phi_1^a - \widehat\phi_0^a )  } . $$
This simply takes the ratio of the corresponding estimators for the effects of $Z$ on $A$ and $Y$, respectively. 

The next result describes the asymptotic properties of the estimator $\widehat\theta$, and gives conditions under which it is $\sqrt{n}$-consistent and asymptotically normal, akin to the earlier Theorem \ref{trtest} for a general $\widehat\bpsi_z$.

\begin{theorem} \label{ivest}
Assume $\| \widehat\phi_z^t - \phi_z^t \| = o_\Pb(1)$ for $z \in \{0,1\}$ and $t \in \{a,y\}$, and $\Pb( \epsilon < \widehat\pi < 1-\epsilon)=\Pb(\epsilon < \widehat\gamma_z < 1-\epsilon)=\Pb\{\Pn( \widehat\phi_1^a - \widehat\phi_0^a ) > \epsilon\}=1$.  
Define
$$ S_{n,z} = \| \widehat\pi - \pi \| \| \widehat\lambda_z - \lambda_z \| + \left( \max_{t \in \{a,y\}} \| \widehat\bbeta_z^t - \bbeta_z^t \| + \| \widehat\gamma_z - \gamma_z \| \right) \| \widehat\gamma_z - \gamma_z \| . $$
Then
$$ \widehat\theta - \theta = O_\Pb\left(\frac{1}{\sqrt{n}} + S_{n,0} + S_{n,1}  \right)  , $$
and if $S_{n,0} + S_{n,1} =o_\Pb(1/\sqrt{n})$, we have
$$ \sqrt{n}(\widehat\theta - \theta) \indist N\left( 0, \var\left\{ \frac{ ( \phi_1^y - \phi_0^y ) - \theta ( \phi_1^a - \phi_0^a )}{\Pb( \phi_1^a - \phi_0^a )}  \right\} \right).  $$
\end{theorem}
\begin{proof}
Note that we have
\begin{align*}
\widehat\theta - \theta &= \frac{ \Pn( \widehat\phi_1^y - \widehat\phi_0^y ) }{ \Pn( \widehat\phi_1^a - \widehat\phi_0^a )  }  - \frac{ \Pb( \phi_1^y - \phi_0^y ) }{ \Pb( \phi_1^a - \phi_0^a )  }  \\
&= \frac{ 1 }{ \Pn( \widehat\phi_1^a - \widehat\phi_0^a )  } \left[ \Big\{ \Pn( \widehat\phi_1^y - \widehat\phi_0^y ) - \Pb( \phi_1^y - \phi_0^y ) \Big\} - \theta \Big\{ \Pn( \widehat\phi_1^a - \widehat\phi_0^a ) - \Pb( \phi_1^a - \phi_0^a ) \Big\} \right] \\
&= \Pn\left\{ \frac{ ( \phi_1^y - \phi_0^y ) - \theta ( \phi_1^a - \phi_0^a )}{\Pb( \phi_1^a - \phi_0^a )} \right\} + o_\Pb(1/\sqrt{n}) \\
& \hspace{.3in} + O_\Pb\left( \| \widehat\pi - \pi \| \max_z \| \widehat\lambda_z - \lambda_z \| + \max_{z,t} \left\{ \left( \| \widehat\bbeta_z^t - \bbeta_z^t \| + \| \widehat\gamma_z - \gamma_z \| \right) \| \widehat\gamma_z - \gamma_z \| \right\} \right) 
\end{align*}
where the third line follows since 
$$ \Pn(\widehat\phi^t ) - \Pb( \phi^t ) = (\Pn-\Pb) (\widehat\phi^t - \phi^t) + (\Pn-\Pb) \phi^t + \Pb(\widehat\phi^t - \phi^t) $$
with the first term $o_\Pb(1/\sqrt{n})$ by Lemma \ref{splitlem} and the third remainder term from Theorem \ref{trtest}, and since $\Pn( \widehat\phi_1^a - \widehat\phi_0^a )$ is bounded away from zero with 
\begin{align*} 
\Pn( \widehat\phi_1^a - \widehat\phi_0^a ) - \Pb( \phi_1^a - \phi_0^a ) &=  (\Pn-\Pb)( \widehat\phi_1^a - \widehat\phi_0^a )  + \Pb\{ ( \widehat\phi_1^a - \widehat\phi_0^a ) - ( \phi_1^a - \phi_0^a )\} \\
&= O_\Pb(1/\sqrt{n}) + \max_z \| \widehat\phi_z^a - \phi_z^a \| = o_\Pb(1)
\end{align*}
where the second equality follows from Lemma \ref{splitlem} and the central limit theorem.
\end{proof}

As before, the estimator $\widehat\theta$ has a fast convergence rate that is second-order involving products of nuisance errors, so that under for example $n^{-1/4}$-type rates the estimator will be $\sqrt{n}$-consistent, asymptotically normal, and efficient. It is also doubly robust, as pointed out in the next corollary.

\begin{corollary} \label{ivconsistency}
Under the conditions of Theorem \ref{ivest}, the estimator $\widehat\theta$ is consistent if either
\begin{enumerate}
\item $\| \widehat\gamma_z - \gamma_z \|=o_\Pb(1)$ for $z \in \{0,1\}$ and $\| \widehat\pi - \pi \| = o_\Pb(1)$, or
\item $\| \widehat\gamma_z - \gamma_z \|=o_\Pb(1)$ and $\| \widehat\lambda_z - \lambda_z \|=o_\Pb(1)$, for $z \in \{0,1\}$.
\end{enumerate}
\end{corollary}

To summarize, the above results extend the work of \citet{mogstad2012instrumental, chaudhuri2016gmm}, and \citet{kennedy2017paradoxes}, by providing an efficient nonparametric estimator of the instrumental variable estimand when some instrument values are missing, allowing adjustment for complex confounding via flexible data-adaptive estimators of the nuisance functions. 

\section{Discussion}

In this paper we filled a gap in the literature by considering nonparametric identification, efficiency theory, and estimation of average treatment effects in the presence of complex confounding and missing exposure values, where the exposure missingness can depend not only on the covariates but also the outcome information. We derived the efficient influence function for the average treatment effect and corresponding nonparametric efficiency bounds, and constructed nonparametric estimators can attain these efficiency bounds under weak rate conditions on the nuisance estimators. This allows one to incorporate modern flexible regression and machine learning tools. We also apply our general results to the problem of causal inference with a partially missing instrumental variable, yielding a new estimator and efficiency bound in this problem as well. 

There are several important avenues for future work. First, it will be useful to study finite-sample properties of the estimators proposed here, in comparison to the more parametric estimators proposed in earlier work. Relatedly, it would be useful to construct an efficient plug-in estimator using targeted maximum likelihood \citep{van2006targeted, van2011targeted}, which would respect bounds on the parameter space, e.g., when $Y$ is bounded. Second, we restricted study to possibly multi-valued but discrete point treatments; it would be of interest to extend to treatments that are continuous \citep{diaz2012population, kennedy2017nonparametric} or time-varying \citep{robins2000marginal, kennedy2019nonparametric}. This would also be useful for continuous instrumental variable problems \citep{kennedy2019robust} with instrument missingness. Further, identification, efficiency theory, and estimation are all more complicated in settings where there is simultaneous missingness in covariates, treatment, and outcome \citep{sun2018inverse}; however, this also occurs often in practice and deserves deeper investigation. \black{Lastly, we assumed exchangeability in the sense of the missing indicator $R$ being conditionally independent of the underlying exposure $Z$ given both covariates $\mathbf{X}$ and  outcome $\mathbf{Y}$; it would be of interest to consider the case where we only assume $R \ind Z \mid \mathbf{X}$. However, there average treatment effects are no longer point identified, and so one would need to consider bounds and/or sensitivity analysis.}

\section{Acknowledgements}

The author thanks Matteo Bonvini, Dylan Small, Mike Daniels, and Joe Hogan for helpful discussions, as well as an anonymous reviewer of a previous version of the manuscript.

\section*{References}
\vspace{-1cm}
\bibliographystyle{abbrvnat}
\bibliography{/Volumes/flashdrive/research/bibliography}

\begin{thebibliography}{35}
\providecommand{\natexlab}[1]{#1}
\providecommand{\url}[1]{\texttt{#1}}
\expandafter\ifx\csname urlstyle\endcsname\relax
  \providecommand{\doi}[1]{doi: #1}\else
  \providecommand{\doi}{doi: \begingroup \urlstyle{rm}\Url}\fi

\bibitem[Abadie(2003)]{abadie2003semiparametric}
A.~Abadie.
\newblock Semiparametric instrumental variable estimation of treatment response
  models.
\newblock \emph{Journal of Econometrics}, 113\penalty0 (2):\penalty0 231--263,
  2003.

\bibitem[Ahn et~al.(2011)Ahn, Mukherjee, Gruber, and Sinha]{ahn2011missing}
J.~Ahn, B.~Mukherjee, S.~B. Gruber, and S.~Sinha.
\newblock Missing exposure data in stereotype regression model: application to
  matched case--control study with disease subclassification.
\newblock \emph{Biometrics}, 67\penalty0 (2):\penalty0 546--558, 2011.

\bibitem[Bickel et~al.(1993)Bickel, Klaassen, Ritov, and
  Wellner]{bickel1993efficient}
P.~J. Bickel, C.~A. Klaassen, Y.~Ritov, and J.~A. Wellner.
\newblock \emph{Efficient and Adaptive Estimation for Semiparametric Models}.
\newblock Baltimore: Johns Hopkins University Press, 1993.

\bibitem[Burgess et~al.(2011)Burgess, Seaman, Lawlor, Casas, and
  Thompson]{burgess2011missing}
S.~Burgess, S.~Seaman, D.~A. Lawlor, J.~P. Casas, and S.~G. Thompson.
\newblock Missing data methods in {Mendelian} randomization studies with
  multiple instruments.
\newblock \emph{American Journal of Epidemiology}, 174\penalty0 (9):\penalty0
  1069--1076, 2011.

\bibitem[Chaudhuri and Guilkey(2016)]{chaudhuri2016gmm}
S.~Chaudhuri and D.~K. Guilkey.
\newblock {GMM} with multiple missing variables.
\newblock \emph{Journal of Applied Econometrics}, 31\penalty0 (4):\penalty0
  678--706, 2016.

\bibitem[Chernozhukov et~al.(2016)Chernozhukov, Chetverikov, Demirer, Duflo,
  Hansen, Newey, and Robins]{chernozhukov2016double}
V.~Chernozhukov, D.~Chetverikov, M.~Demirer, E.~Duflo, C.~Hansen, W.~Newey, and
  J.~M. Robins.
\newblock Double machine learning for treatment and causal parameters.
\newblock \emph{arXiv preprint arXiv:1608.00060}, 2016.

\bibitem[D{\'\i}az and {van der Laan}(2012)]{diaz2012population}
I.~D{\'\i}az and M.~J. {van der Laan}.
\newblock Population intervention causal effects based on stochastic
  interventions.
\newblock \emph{Biometrics}, 68\penalty0 (2):\penalty0 541--549, 2012.

\bibitem[Farrell(2015)]{farrell2015robust}
M.~H. Farrell.
\newblock Robust inference on average treatment effects with possibly more
  covariates than observations.
\newblock \emph{Journal of Econometrics}, 189\penalty0 (1):\penalty0 1--23,
  2015.

\bibitem[Hern{\'a}n and Robins(2006)]{hernan2006instruments}
M.~A. Hern{\'a}n and J.~M. Robins.
\newblock Instruments for causal inference: an epidemiologist's dream?
\newblock \emph{Epidemiology}, 17\penalty0 (4):\penalty0 360--372, 2006.

\bibitem[Imbens(2004)]{imbens2004nonparametric}
G.~W. Imbens.
\newblock Nonparametric estimation of average treatment effects under
  exogeneity: A review.
\newblock \emph{Review of Economics and Statistics}, 86\penalty0 (1):\penalty0
  4--29, 2004.

\bibitem[Imbens and Angrist(1994)]{imbens1994identification}
G.~W. Imbens and J.~D. Angrist.
\newblock Identification and estimation of local average treatment effects.
\newblock \emph{Econometrica}, 62\penalty0 (2):\penalty0 467--475, 1994.

\bibitem[Kennedy(2019)]{kennedy2019nonparametric}
E.~H. Kennedy.
\newblock Nonparametric causal effects based on incremental propensity score
  interventions.
\newblock \emph{Journal of the American Statistical Association}, 114\penalty0
  (526):\penalty0 645--656, 2019.

\bibitem[Kennedy and Small(2017)]{kennedy2017paradoxes}
E.~H. Kennedy and D.~S. Small.
\newblock Paradoxes in instrumental variable studies with missing data and
  one-sided noncompliance.
\newblock \emph{Journal of the French Statistical Society (to appear)}, 2017.

\bibitem[Kennedy et~al.(2017)Kennedy, Ma, McHugh, and
  Small]{kennedy2017nonparametric}
E.~H. Kennedy, Z.~Ma, M.~D. McHugh, and D.~S. Small.
\newblock Nonparametric methods for doubly robust estimation of continuous
  treatment effects.
\newblock \emph{Journal of the Royal Statistical Society: Series B},
  79\penalty0 (4):\penalty0 1229--1245, 2017.

\bibitem[Kennedy et~al.(2019{\natexlab{a}})Kennedy, Balakrishnan, and
  G'Sell]{kennedy2019sharp}
E.~H. Kennedy, S.~Balakrishnan, and M.~G'Sell.
\newblock Sharp instruments for classifying compliers and generalizing causal
  effects.
\newblock \emph{The Annals of Statistics (to appear)}, 2019{\natexlab{a}}.

\bibitem[Kennedy et~al.(2019{\natexlab{b}})Kennedy, Lorch, and
  Small]{kennedy2019robust}
E.~H. Kennedy, S.~Lorch, and D.~S. Small.
\newblock Robust causal inference with continuous instruments using the local
  instrumental variable curve.
\newblock \emph{Journal of the Royal Statistical Society: Series B},
  81\penalty0 (1):\penalty0 121--143, 2019{\natexlab{b}}.

\bibitem[Mebane~Jr and Poast(2013)]{mebane2013causal}
W.~R. Mebane~Jr and P.~Poast.
\newblock Causal inference without ignorability: Identification with nonrandom
  assignment and missing treatment data.
\newblock \emph{Political Analysis}, 21\penalty0 (2):\penalty0 233--251, 2013.

\bibitem[Mogstad and Wiswall(2012)]{mogstad2012instrumental}
M.~Mogstad and M.~Wiswall.
\newblock Instrumental variables estimation with partially missing instruments.
\newblock \emph{Economics Letters}, 114\penalty0 (2):\penalty0 186--189, 2012.

\bibitem[Molinari(2010)]{molinari2010missing}
F.~Molinari.
\newblock Missing treatments.
\newblock \emph{Journal of Business \& Economic Statistics}, 28\penalty0
  (1):\penalty0 82--95, 2010.

\bibitem[Robins(2000)]{robins2000robust}
J.~M. Robins.
\newblock Robust estimation in sequentially ignorable missing data and causal
  inference models.
\newblock \emph{Proceedings of the American Statistical Association},
  1999:\penalty0 6--10, 2000.

\bibitem[Robins et~al.(1994)Robins, Rotnitzky, and Zhao]{robins1994estimation}
J.~M. Robins, A.~Rotnitzky, and L.~P. Zhao.
\newblock Estimation of regression coefficients when some regressors are not
  always observed.
\newblock \emph{Journal of the American Statistical Association}, 89\penalty0
  (427):\penalty0 846--866, 1994.

\bibitem[Robins et~al.(2000)Robins, Hernan, and Brumback]{robins2000marginal}
J.~M. Robins, M.~A. Hernan, and B.~Brumback.
\newblock Marginal structural models and causal inference in epidemiology.
\newblock \emph{Epidemiology}, 11\penalty0 (5):\penalty0 550--560, 2000.

\bibitem[Robins et~al.(2008)Robins, Li, {Tchetgen Tchetgen}, and {van der
  Vaart}]{robins2008higher}
J.~M. Robins, L.~Li, E.~J. {Tchetgen Tchetgen}, and A.~W. {van der Vaart}.
\newblock Higher order influence functions and minimax estimation of nonlinear
  functionals.
\newblock \emph{Probability and Statistics: Essays in Honor of David A.
  Freedman}, pages 335--421, 2008.

\bibitem[Scharfstein et~al.(1999)Scharfstein, Rotnitzky, and
  Robins]{scharfstein1999adjusting}
D.~O. Scharfstein, A.~Rotnitzky, and J.~M. Robins.
\newblock Adjusting for nonignorable drop-out using semiparametric nonresponse
  models.
\newblock \emph{Journal of the American Statistical Association}, 94\penalty0
  (448):\penalty0 1096--1120, 1999.

\bibitem[Shardell and Hicks(2014)]{shardell2014statistical}
M.~Shardell and G.~E. Hicks.
\newblock Statistical analysis with missing exposure data measured by proxy
  respondents: a misclassification problem within a missing-data problem.
\newblock \emph{Statistics in Medicine}, 33\penalty0 (25):\penalty0 4437--4452,
  2014.

\bibitem[Shortreed and Forbes(2010)]{shortreed2010missing}
S.~M. Shortreed and A.~B. Forbes.
\newblock Missing data in the exposure of interest and marginal structural
  models: a simulation study based on the framingham heart study.
\newblock \emph{Statistics in Medicine}, 29\penalty0 (4):\penalty0 431--443,
  2010.

\bibitem[Sun and {Tchetgen Tchetgen}(2018)]{sun2018inverse}
B.~Sun and E.~J. {Tchetgen Tchetgen}.
\newblock On inverse probability weighting for nonmonotone missing at random
  data.
\newblock \emph{Journal of the American Statistical Association}, 113\penalty0
  (521):\penalty0 369--379, 2018.

\bibitem[Tsiatis(2006)]{tsiatis2006semiparametric}
A.~A. Tsiatis.
\newblock \emph{Semiparametric Theory and Missing Data}.
\newblock New York: Springer, 2006.

\bibitem[{van der Laan} and Robins(2003)]{van2003unified}
M.~J. {van der Laan} and J.~M. Robins.
\newblock \emph{Unified Methods for Censored Longitudinal Data and Causality}.
\newblock New York: Springer, 2003.

\bibitem[{van der Laan} and Rose(2011)]{van2011targeted}
M.~J. {van der Laan} and S.~Rose.
\newblock \emph{Targeted Learning: Causal Inference for Observational and
  Experimental Data}.
\newblock Springer, 2011.

\bibitem[{van der Laan} and Rubin(2006)]{van2006targeted}
M.~J. {van der Laan} and D.~B. Rubin.
\newblock Targeted maximum likelihood learning.
\newblock \emph{UC Berkeley Division of Biostatistics Working Paper Series},
  212, 2006.

\bibitem[{van der Vaart}(2002)]{van2002semiparametric}
A.~W. {van der Vaart}.
\newblock Semiparametric statistics.
\newblock \emph{In: Lectures on Probability Theory and Statistics}, pages
  331--457, 2002.

\bibitem[Williamson et~al.(2012)Williamson, Forbes, and
  Wolfe]{williamson2012doubly}
E.~Williamson, A.~Forbes, and R.~Wolfe.
\newblock Doubly robust estimators of causal exposure effects with missing data
  in the outcome, exposure or a confounder.
\newblock \emph{Statistics in Medicine}, 31\penalty0 (30):\penalty0 4382--4400,
  2012.

\bibitem[Zhang et~al.(2016)Zhang, Liu, Zhang, Tang, and Zhang]{zhang2016causal}
Z.~Zhang, W.~Liu, B.~Zhang, L.~Tang, and J.~Zhang.
\newblock Causal inference with missing exposure information: methods and
  applications to an obstetric study.
\newblock \emph{Statistical Methods in Medical Research}, 25\penalty0
  (5):\penalty0 2053--2066, 2016.

\bibitem[Zheng and {van der Laan}(2010)]{zheng2010asymptotic}
W.~Zheng and M.~J. {van der Laan}.
\newblock Asymptotic theory for cross-validated targeted maximum likelihood
  estimation.
\newblock \emph{UC Berkeley Division of Biostatistics Working Paper Series},
  Paper 273:\penalty0 1--58, 2010.

\end{thebibliography}

\newpage
\setcounter{page}{1}
\appendix

\section{Appendix}

The following lemma from \citet{kennedy2019sharp}  is useful in proving Theorem \ref{trtest}.

\begin{lemma} \label{splitlem}
Let $\widehat{f}(\bo)$ be a function estimated from a sample $\bO^N=(\bO_{n+1},...,\bO_N)$, and let $\Pn$ denote the empirical measure over $(\bO_1,...,\bO_n)$, which is independent of $\bO^N$. Then 
$$ (\Pn-\Pb) (\widehat{f}-f) = O_\Pb\left( \frac{ \| \widehat{f}-f \| }{\sqrt{n}}  \right) . $$ 
\end{lemma}

\begin{proof}
First note that, conditional on $\bO^N$, the term in question has mean zero since
$$ \E\Big\{ \Pn(\widehat{f}-f) \Bigm| \bO^N \Big\}  = \E(\widehat{f}-f \mid \bO^N) = \Pb(\widehat{f}-f) . $$
The conditional variance is
\begin{align*}
\var\Big\{ (\Pn-\Pb) (\widehat{f}-f) \Bigm| \bO^N \Big\} &=  \var\Big\{ \Pn(\widehat{f}-f) \Bigm| \bO^N \Big\} = \frac{1}{n} \var(\widehat{f}-f \mid  \bO^N ) \leq \|\widehat{f}-f\|^2 /n .
\end{align*}
Therefore using Chebyshev's inequality we have
\begin{align*}
\Pb\left\{ \frac{ | (\Pn-\Pb)(\widehat{f}-f) | }{ \| \widehat{f}-f \| / \sqrt{n} } \geq t \right\} &= \E\left[ \Pb\left\{ \frac{ | (\Pn-\Pb)(\widehat{f}-f) | }{ \| \widehat{f}-f \| / \sqrt{n} } \geq t \Bigm| \bO^N \right\} \right] \leq \frac{1}{t^2} .
\end{align*}
Thus for any $\epsilon>0$ we can pick $t=1/\sqrt{\epsilon}$ so that the probability above is no more than $\epsilon$, which yields the result.
\end{proof}

\end{document}